\documentclass[runningheads]{llncs}
\usepackage{graphicx}
\usepackage{xcolor}
\usepackage{multirow}
\usepackage{algorithm}
\usepackage{algpseudocode}
\algnewcommand\algorithmicforeach{\textbf{for each}}
\algdef{S}[FOR]{ForEach}[1]{\algorithmicforeach\ #1\ \algorithmicdo}

\usepackage{cite}
\usepackage{url}

\usepackage[shortlabels]{enumitem}


\usepackage{subcaption}

\usepackage{hyperref}
\hypersetup{backref=true,
    pdftex=true,
    pagebackref=true,               
    hyperindex=true,                
    colorlinks=true,                
    breaklinks=true,                
    urlcolor= black,                
    linkcolor= blue,                
    bookmarks=true,                 
    bookmarksopen=false,
    filecolor=black,
    citecolor=blue,
    linkbordercolor=blue
}
\AtBeginDocument{\hypersetup{pdfborder={0 0 1}}}
\usepackage{hypcap}

\usepackage{algpseudocode}
\algdef{SE}[DOWHILE]{Do}{doWhile}{\algorithmicdo}[1]{\algorithmicwhile\ #1}%

\algnewcommand{\IIf}[1]{\State\algorithmicif\ #1\ \algorithmicthen}
\algnewcommand{\EndIIf}{\unskip\ \algorithmicend\ \algorithmicif}
\algdef{SE}[DOWHILE]{Do}{doWhile}{\algorithmicdo}[1]{\algorithmicwhile\ #1}%

\usepackage{amsmath,wasysym,amssymb}
\usepackage{cite}

\begin{document}

\title{Towards a Complete Direct Mapping From Relational Databases To Property Graphs}
\titlerunning{Complete Direct Mapping From Relational Databases To Property Graphs}
%
\author{Abdelkrim Boudaoud , Houari Mahfoud , Azeddine Chikh}
\authorrunning{A. Boudaoud et al.}
%
\institute{Abou-Bekr Belkaid University \& LRIT Laboratory\\Tlemcen, Algeria\\
\email{\{abdelkrim.boudaoud, houari.mahfoud, azeddine.chikh\}@univ-tlemcen.dz}}
\maketitle              
\begin{abstract}
It is increasingly common to find complex data represented through the graph model. Contrary to relational models, graphs offer a high capacity for executing analytical tasks on complex data. Since a huge amount of data is still presented in terms of relational tables, it is necessary to understand how to translate this data into graphs. This paper proposes a \emph{complete mapping} process that allows transforming any relational database (schema and instance) into a property graph database (schema and instance). Contrary to existing mappings, our solution preserves the three fundamental mapping properties, namely: \emph{information preservation, semantic preservation} and \emph{query preservation}. Moreover, we study mapping any \emph{SQL} query into an equivalent \emph{Cypher} query, which makes our solution practical. Existing solutions are either incomplete or based on non-practical query language. Thus, this work is the first complete and practical solution for mapping relations to graphs.

\keywords{Direct mapping, Complete mapping, Relational database, Graph database, SQL, Cypher}
\end{abstract}

\section{Introduction}\label{section:introduction}
Relational databases (RDs) have been widely used and studied by researchers and practitioners for decades due to their simplicity, low data redundancy, high data consistency, and uniform query language (SQL).
Hence, the size of web data has grown exponentially during the last two decades. The interconnections between web data entities (e.g. interconnection between YouTube videos or people on Facebook) are measured by billions or even trillions \cite{FacebookMetrics} which pushes the relational model to quickly reach its limits as querying high interconnected web data requires complex SQL queries which are time-consuming. To overcome this limit, the graph database model is increasingly used on the Web due to its flexibility to present data in a normal form, its efficiency to query a huge amount of data and its analytic powerful.
This suggests studying a mapping from RDs to graph databases (GDs) to benefit from the aforementioned advantages.
This kind of mapping has not received more attention from researchers since only a few works \cite{Stoica2019,Stoica2019a,DeVirgilio2013,DeVirgilio2014} have considered it. A real-life example of this mapping has been discussed in \cite{Stoica2019}: ``\emph{investigative journalists have recently found, through graph analytics, surprising social relationships between executives of companies within the Offshore Leaks financial social network data set, linking company officers and their companies registered in the Bahamas. The Offshore Leaks PG was constructed as a mapping from relational database
(RDB) sources}''. In a nutshell, the proposed mappings suffer from at least one of the following limits : a) they do not study fundamental properties of mapping; b) they do not consider a practical query language to make the approach more useful; c) they generate an obfuscated schema.

This paper aims to provide a \emph{complete mapping} (\emph{CM}) from RDs to GDs by investigating the fundamental properties of mapping \cite{Stoica2019a}, namely : \emph{ information preservation} (IP), \emph{query preservation} (QP), and \emph{semantic preservation} (SP). In addition to data mapping, we study the mapping of SQL queries to Cypher queries which makes our results more practical since SQL and Cypher are the most used query languages for relational and graph data respectively.

\noindent\textbf{Contributions and Road-map.} Our main contributions are as follows :
\textit{1)} we formalize a \emph{CM} process that maps RDs to GDs in the presence of schema; \textit{2)} we propose definitions of \emph{schema graph} and \emph{graph consistency} that are necessary for this mapping; \textit{3)} we show that our \emph{CM} preserves the three fundamental mapping properties (\emph{IP}, \emph{SM}, and \emph{QP}); \textit{4)} in order to prove \emph{QP}, we propose an algorithm to map SQL queries into equivalent Cypher queries. To our knowledge, this work is the first complete effort that investigates mapping relations to graphs.

\noindent\textbf{Related Work.} We classify previous works as follows:\\
\noindent\textit{\underline{Mapping RDs to RDF Data.}} Squeda et al. \cite{Sequeda2012} studied the mapping of RDs to RDF graph data and relational schema to OWL ontology's. Moreover, SQL queries are translated into SPARQL queries. They were the first to define a set of mapping properties : information preservation, query preservation, semantic preservation and monotonicity preservation. They proved first that their mapping is information preserving and query preserving, while when it comes to the two remaining properties, preserving semantics makes the mapping not monotonicity preserving.\\
\noindent\textit{\underline{Mapping RDs to Graph Data.}} De Virgilio et al. \cite{DeVirgilio2013,DeVirgilio2014} studied mapping a) RDs to property graph data (PG) by considering schema both in source and target; and b) any SQL query into a set of graph traversal operations that realize the same semantic over the resulted graph data.We remark that the proposed mapping obfuscates the relational schema since the resulted graph schema is difficult to understand. Moreover, the mapping does not consider typed data. From the practical point of view, the graph querying language considered is not really used in practice and the proposed query mapping depends on the syntax and semantics of this language which makes hard the application of their proposal for another query language.
In addition, they apply an aggregation process that maps different relational tuples to the same graph vertex in order to optimize graph traversal operations. However, this makes the  mapping not information preserving and can skew the result of some analytical tasks that one would like to apply over the resulted data graph.

Stoica et al.\cite{Stoica2019,Stoica2019a} studied the mapping of RDs to GDs and any relational query (formalized as an extension of relational algebra) into a G-core query. Firstly, the choice of source and destination languages hinders the practicability of the approach. Moreover, it is hard to see if the mapping is semantic preserving since no definition of graph data consistency is given. Attributes, primary and foreign keys are verbosely represented by the data graph, which makes this later hard to understand and to query.

O.Orel et al.  \cite{Orel2017} discussed mapping relational data only into property graphs without giving attention neither to schema nor mapping properties.

The Neo4j system provides a tool called Neo4j-ETL \cite{neo4j}, which allows users to import their relational data into Neo4j to be presented as property graphs. Notice that the relational structure (both instance and schema) is not preserved during this mapping since some tuples of the relational data (resp. relations of the relational schema) are represented as edges for storage concerns (as done in \cite{DeVirgilio2014}). However, as remarked in \cite{Stoica2019}, this may skew the results of some analytical tasks (e.g. density of the generated graphs). Moreover, Neo4j-ETL does not allow the mapping of queries. S. Li et al. \cite{Li2021} study an extension of Neo4j-ETL by proposing mapping of SQL queries to Cypher queries. However, their mapping inherits the limits of Neo4j-ETL. In addition, no detailed algorithm is given for the query mapper which makes impossible the comparison of their proposal with other ones. This is also the limit of \cite{Matsumoto2018}.

Finally, Angles et al. \cite{Angles2020} studied mappings RDF databases to property graphs by considering both data and schema. They proved that their mapping ensures both information and semantic preservation properties.

Table ~\ref{tcw} summarizes most important features of related works.

\begin{table}[t]\centering\caption{Comparative table of related works.}\label{tcw}\begin{tabular}{lllllll|clcl|clclcl|clll|ccccclclllll}\cline{1-21}\multicolumn{2}{|c|}{\multirow{2}{*}{\bf Type}} & \multicolumn{5}{c|}{\multirow{2}{*}{\bf Work}} & \multicolumn{4}{c|}{\bf Mapping} & \multicolumn{6}{c|}{\begin{tabular}[c]{@{}c@{}}\bf Preserved \\ \bf properties\end{tabular}} & \multicolumn{4}{c|}{\multirow{2}{*}{\bf Mapping rules}} & & & & & \multicolumn{4}{c}{} & & & & \\ \cline{8-17}\multicolumn{2}{|c|}{} & \multicolumn{5}{c|}{} & \multicolumn{2}{c|}{\bf Schema} & \multicolumn{2}{c|}{ \bf Instance} & \multicolumn{2}{c|}{\bf IP} & \multicolumn{2}{c|}{\bf QP} & \multicolumn{2}{c|}{\bf SP} & \multicolumn{4}{c|}{} & & & & & \multicolumn{2}{c}{} & \multicolumn{2}{c}{} & \multicolumn{4}{c}{} \\ \cline{1-21}\multicolumn{2}{|l|}{\multirow{3}{*}{RDs $\rightarrow$ RDF}} & \multicolumn{5}{l|}{Sequeda et al. \cite{Sequeda2012}} & \multicolumn{2}{c|}{\checkmark} & \multicolumn{2}{c|}{\checkmark} & \multicolumn{2}{c|}{\checkmark} & \multicolumn{2}{c|}{\checkmark} & \multicolumn{2}{c|}{\checkmark} & \multicolumn{4}{c|}{\checkmark} & \multicolumn{4}{c}{} & \multicolumn{2}{c}{} & \multicolumn{2}{c}{} & & & & \\ \cline{3-21}\multicolumn{2}{|l|}{} & \multicolumn{5}{l|}{De vergillio et al. \cite{DeVirgilio2013,DeVirgilio2014}} & \multicolumn{2}{c|}{\checkmark} & \multicolumn{2}{c|}{\checkmark} & \multicolumn{2}{c|}{} & \multicolumn{2}{c|}{} & \multicolumn{2}{c|}{} & \multicolumn{4}{c|}{\checkmark} & \multicolumn{4}{c}{} & \multicolumn{2}{c}{} & \multicolumn{2}{c}{} & & & & \\ \cline{3-21}\multicolumn{2}{|l|}{} & \multicolumn{5}{l|}{Stoica et al. \cite{Stoica2019,Stoica2019a}} & \multicolumn{2}{c|}{\checkmark} & \multicolumn{2}{c|}{\checkmark} & \multicolumn{2}{c|}{\checkmark} & \multicolumn{2}{c|}{\checkmark} & \multicolumn{2}{c|}{\checkmark} & \multicolumn{4}{c|}{\checkmark} & \multicolumn{4}{c}{} & \multicolumn{2}{c}{} & \multicolumn{2}{c}{} & & & & \\ \cline{1-21}\multicolumn{2}{|l|}{RDF $\rightarrow$ PG} & \multicolumn{5}{l|}{Angeles et al. \cite{Angles2020}} & \multicolumn{2}{c|}{\checkmark} & \multicolumn{2}{c|}{\checkmark} & \multicolumn{2}{c|}{\checkmark} & \multicolumn{2}{c|}{} & \multicolumn{2}{c|}{\checkmark} & \multicolumn{4}{c|}{\checkmark} & \multicolumn{4}{c}{} & \multicolumn{2}{c}{} & \multicolumn{2}{c}{} & & & & \\ \cline{1-21}\multicolumn{2}{|l|}{\multirow{3}{*}{RDs $\rightarrow$ PG}} & \multicolumn{5}{l|}{O.Orel et al. \cite{Orel2017}} & \multicolumn{2}{c|}{} & \multicolumn{2}{c|}{\checkmark} & \multicolumn{2}{c|}{} & \multicolumn{2}{c|}{} & \multicolumn{2}{c|}{} & \multicolumn{4}{c|}{\checkmark} & \multicolumn{4}{c}{} & \multicolumn{2}{c}{} & \multicolumn{2}{c}{} & & & & \\ \cline{3-21}\multicolumn{2}{|l|}{} & \multicolumn{5}{l|}{S.Li et al. \cite{Li2021}} & \multicolumn{2}{l|}{} & \multicolumn{2}{c|}{\checkmark} & \multicolumn{2}{l|}{} & \multicolumn{2}{l|}{} & \multicolumn{2}{l|}{} & \multicolumn{4}{l|}{} & \multicolumn{1}{l}{} & \multicolumn{1}{l}{} & \multicolumn{1}{l}{} & \multicolumn{1}{l}{} & \multicolumn{1}{l}{} & & \multicolumn{1}{l}{} & & & & & \\ \cline{3-21}\multicolumn{2}{|l|}{} & \multicolumn{5}{l|}{Our Work} & \multicolumn{2}{c|}{\checkmark} & \multicolumn{2}{c|}{\checkmark} & \multicolumn{2}{c|}{\checkmark} & \multicolumn{2}{c|}{\checkmark} & \multicolumn{2}{c|}{\checkmark} & \multicolumn{4}{c|}{\checkmark} & & & & & & \multicolumn{1}{c}{} & & \multicolumn{1}{c}{} & & & & \\ \cline{1-21}\end{tabular}\end{table}

\section{Preliminaries}
This section defines the several notions that will be used throughout this paper.\\

Let $\mathcal{R}$ be an infinite set of relation names, $\mathcal{A}$ is an infinite set of attribute names with a special attribute $tid$, $\mathcal{T}$ is a finite set of attribute types (\emph{String}, \emph{Date}, \emph{Integer}, \emph{Float}, \emph{Boolean}, \emph{Object}), $\mathcal{D}$ is a countably infinite domain of data values with a special value \textsc{null}.

\subsection{Relational Databases}

 A \emph{relational schema} is a tuple $S=(R, A, T, \Sigma)$  where:

\begin{enumerate}

\item $R \subseteq \mathcal{R}$ is a finite set of relation names;
\item $A$ is a function assigning a finite set of attributes for each relation $r\in R$ such that $A(r) \subseteq \mathcal{A}\backslash \{tid\}$;
\item $T$ is a function assigning a type for each attribute of a relation, i.e. for each $r\in R$ and each $a \in A(r) \setminus \{tid\}$, $T(a) \subseteq \mathcal{T}$;

\item $\Sigma$  is a finite set of \emph{primary} (PKs) and \emph{foreign} keys (FKs) defined over $R$ and $A$. A \emph{primary key} over a relation $r \in R$ is an expression of the form $r[a_{1},\cdots,a_{n}]$ where $a_{1\leq i\leq n}\in A(r)$. A \emph{foreign key} over two relations
$r$ and $s$ is an expression of the form $r[a_{1},\cdots,a_{n}] \to s[b_{1},\cdots,b_{n}]$ where  $a_{1\leq i\leq n}\in A(r)$ and
$s[b_{1},\cdots,b_{n}] \in \Sigma.$

\end{enumerate}

An \emph{instance} $I$ of $S$ is an assignment to each $r \in R$ of a finite set  $I(r) = \{t_{1},\cdots, t_{n}\}$ of \emph{tuples}. Each \emph{tuple} $t_{i}:A(r)\cup\{tid\}\rightarrow\mathcal{D}$ is identified by $tid\neq\textsc{null}$ and assigns a value to each attribute $a\in A(r)$. We use $t_i(a)$ (resp. $t_i(tid)$) to refer to the value of attribute $a$ (resp. $tid$) in tuple $t_i$. Moreover, for any tuples $t_i,t_j\in I(r)$, $t_i(tid)\neq t_j(tid)$ if $i\neq j$.

For any instance $I$ of a relational schema $S=(R,A,T,\Sigma)$, we say that $I$ satisfies a primary key $r[a_{1},\cdots,a_{n}]$ in $\Sigma$  if: 1) for each tuple $t\in I(r)$, $t(a_{_{1\leq i\leq n}})\neq\textsc{null}$; and 2) for any $t^{'}\in I(r)$, if $t(a_{_{1\leq i\leq n}})=t^{'}(a_{_{1\leq i\leq n}})$ then $t=t^{'}$ must hold. Moreover, $I$ satisfies a foreign key $r[a_{1},\cdots,a_{n}] \to s[b_{1},\cdots,b_{n}]$ in $\Sigma$  if: 1) $I$ satisfies $s[b_{1},\cdots,b_{n}]$; and 2) for each tuple $t\in I(r)$, either $t(a_{_{1\leq i\leq n}})=\textsc{null}$ or there exists a tuple $t^{'}\in I(s)$ where $t(a_{_{1\leq i\leq n}})=t^{'}(b_{_{1\leq i\leq n}})$. The instance $I$ satisfies all integrity constraints in $\Sigma$, denoted by $I\vDash \Sigma$, if it satisfy all primary keys and foreign keys in $\Sigma$.

Finally, a \emph{relational database} is defined with $D_R=(S_R,I_R)$ where $S_R$ is a relational schema and $I_R$ is an instance of $S_R$.

\subsection{Property Graphs}

A \emph{property graph} \emph{(PG)} is a multi-graph structure composed of labeled and attributed vertices and edges defined with $G\!=\!(V,E,L,A)$ where: 1) $V$ is a finite set of vertices; 2) $E \subseteq V \times V$ is a finite set of directed edges where $(v,v^{'})\in E$ is an edge starting at vertex $v$ and ending at vertex $v^{'}$; 3) $L$ is a function that assigns a label to each vertex in $V$ (resp. edge in $E$); and 4) $A$ is a function assigning a nonempty set of key-value pairs to each vertex (resp. edge).
For any edge $e\in E$, we denote by $e.s$ (resp. $e.d$) the starting (resp. ending) vertex of $e$.

\subsection{SQL queries and Cypher queries}

We study in this paper the mapping of relational data into PG data. In addition, we show that any relational query over the source data can be translated into an equivalent graph query over the generated data graph. To this end, we model relational queries with the SQL language \cite{Guagliardo2018}  and the graph queries with the Cypher language \cite{Francis2018}  since each of these languages is the most used in its category. To establish a compromise between the expressive power of our mapping and its processing time, we consider a simple but very practical class of SQL queries and we define its corresponding class of Cypher queries. It is necessary to understand the relations between basic SQL queries and basic Cypher queries before studying all the expressive power of these languages.\\

The well-known syntax of SQL queries is “Select $I$ from $R$ where $C$” where: a) $I$ is a set of  items; b) $R$ is a set of relations names; and c) $C$ is a set of conditions. Intuitively, an SQL query selects first some tuples of relations in $R$ that satisfy conditions in $C$. Then, the values of some records (specified by $I$) of these tuples are returned to the user. 

On the other side, the Cypher queries considered in this paper have the syntax: “$Match$ patterns $Where$ conditions $Return$ items”. Notice that a Cypher query aims to find, based on edge-isomorphism, all subgraphs in some  data graph that match the pattern specified by the query and also satisfy conditions defined by this latter. Once found, only some parts (i.e. vertices, edges, and/or attributes) of these subgraphs are returned, based on items specified by $Return$ clause. Therefore, the $Match$ clause specifies the structure of subgraphs we are looking for on the data graph; the $Where$ clause specifies conditions (based on vertices, edges and/or attributes) these subgraphs must satisfy; and finally, the $Return$ clause returns some parts of these subgraphs to the user as a table. 

\begin{figure}[t]
\centering
\includegraphics[width=12cm]{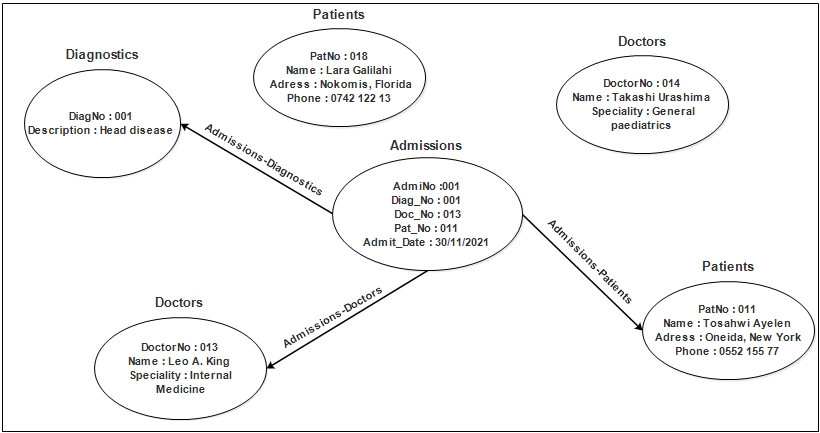}
\caption{Example of data graph.}
\label{EX_PG}
\end{figure}

\begin{example}
Fig.\ref{EX_PG} depicts a data graph where vertices represent entities (i.e. $Doctor$, $Patient$, $Diagnostic$ and $Admission$);  inner information (called attributes) represent properties of this vertex (e.g. $Speciality$); and edges represent relationships between these entities. For instance, an $Admission$ vertex may be connected to some $Patient$, $Doctor$ and  $Diagnostic$ vertices to specify for some patient: a) his doctor; b) information about his admission at the hospital; and c) diagnostics made for this patient. $\square$

The following Cypher query returns the name of each patient who is admitted at some date :\\

\begin{center}
\scriptsize
     \begin{minipage}{11cm}
        \begin{center}
           \begin{flushleft}\small{
		$MATCH \: (a:Admissions)<-[:Admissions-Patients]-(p:Patients)$\\
		$WHERE \:  a.Admi\_date \: = \: "30/11/2021"$\\
		$RETURN \:  p.Name$\\}
\end{flushleft}
        \end{center}
    \end{minipage}
\end{center}
$\square$
\end{example}

\subsection{Direct Mapping (\emph{DM})}

Inspired from\cite{Stoica2019,Stoica2019a,Sequeda2012}, We define in this section the \emph{direct mapping}  from a relational database into a graph database and we discuss its properties. Given a relational database $D_R$ composed of $S_R$ and $I_{R}$, a direct mapping consists of translating each entity in $D_{R}$ into a graph database without any user interaction. That is, any $D_{R}$=($S_{R}$, $I_{R}$) (with possibly empty $S_{R}$),  is translated automatically into a pair of property graphs ($S_{G}$, $I_{G}$) (with possibly empty $S_{G}$), that we call a graph database. Let $\mathcal{D}_{R}$ be an infinite set of relational databases, and $\mathcal{D}_{G}$ be an infinite set of graph databases. Based on these notions, we give the next definition of direct mapping and its properties.

\begin{definition}
A \emph{direct mapping} is a total function $DM:\mathcal{D}_{R}\to \mathcal{D}_{G}$. $\square$
\end{definition}

Intuitively, for each $D_R\in \mathcal{D}_{R}$, $DM(D_R)$ produces a graph database $D_G\in \mathcal{D}_{G}$ that aims to represent the source relational database (i.e. instance and optionally schema) in terms of a graph.

We define next fundamental properties that a direct mapping must preserve  \cite{Sequeda2012}, namely: \emph{information preservation}, \emph{query preservation} and \emph{semantic preservation}. The two first properties ensure that the direct mapping does not lose neither information nor semantic of the relational database being translated. The last property ensures that the mapping does not hinder the querying capabilities as any relational query can be translated into a graph query.

\subsubsection{Information Preservation}

A direct mapping \emph{DM} is \emph{information preserving} if no information is lost during the mapping of any relational database.

\begin{definition}[Information preservation]
A direct mapping \emph{DM} is \emph{information preserving} if there is a computable inverse mapping $DM^{-1}: \mathcal{D}_{G} \to \mathcal{D}_{R}$ satisfying $DM^{-1}(DM(D_{R})) = D_{R}$ for any $D_R\in\mathcal{D}_{R}$. $\square$
\end{definition}

\subsubsection{Query Preservation}

Recall that both SQL and Cypher queries return a result modeled as a table where columns represent entities requested by the query (using Select clause in case of SQL, or Return clause in case of Cypher), while each row assigns values to these entities. In addition, the result of both SQL and Cypher queries may contain repeated rows.

Let $I_{R}$ be a relational instance and $Q_{s}$  be an SQL query expressed over $I_{R}$. We denote by $[Q_{s}]_{I_{R}}$ the result table of $Q_{s}$ over $I_{R}$. Similarly, $[Q_{c}]_{I_{G}}$ is the result table of a Cypher query $Q_{c}$ over an instance graph $I_{G}$. Moreover, we denote by $[Q_{s}]^{*}_{I_{R}}$ (resp. $[Q_{c}]^{*}_{I_{G}}$) the refined table that has no repeated row.

A direct mapping \emph{DM} is \emph{query preserving} if any query over the relational database $D_R$ can be translated into an equivalent query over the graph database $D_G$ that results from the mapping of $D_R$.
That is, query preservation ensures that every relational query can be evaluated using the mapped instance graph.

Since SQL and Cypher languages return results in different forms, proving query preservation consists to define a mapping from the SQL result to the Cypher result. This principle was proposed first in \cite{Stoica2019a} between relational and RDF queries. Therefore, we revise the definition of query preservation as follows:

\begin{definition}[Query preservation]

A direct mapping \emph{DM} is \emph{query preserving} if, for any relational database  $D_{R}$=($S_{R}$,$I_{R}$)  and any SQL query $Q_s$, there exists a Cypher query $Q_{c}$  such that: each row in $[Q_{s}]^{*}_{I_{R}}$ can be mapped into a row in   $[Q_{c}]^{*}_{I_{G} \in DM(D_R)}$ and vice versa. By mapping a row r into a row r', we assume that r and r' contain the same data with possibly different forms. $\square$

\end{definition}

\subsubsection{Semantics Preservation} 

A direct mapping \emph{DM} is \emph{semantics preserving} if any consistent (resp. inconsistent) relational database is translated into a consistent (resp. inconsistent) graph database .


\begin{definition}[Semantic preservation]
A direct mapping \emph{DM} is semantic preserving if, for any relational database $D_R=( S_R,I_R)$ with a set of integrity constraints $\Sigma$,  $I_R  \models \Sigma$ iff: $DM(D_{R})$ produces a consistent graph database. $\square$
\end{definition}

Notice that no previous work have considered semantic preservation over data graphs. That is, no definition of graph database consistency have been given. We shall give later our own definition.


\section{Complete Mapping (\emph{CM})}
In this section, we propose a complete mapping \emph{CM} that transforms a complete relational database (schema and instance) into a complete graph database (schema and instance). We call our mapping \emph{Complete} since some proposed mappings (e.g \cite{Orel2017}) deal only with data and not schema.

\begin{definition}[\emph{Complete Mapping}]
A complete mapping is a function $ CM:\mathcal{D}_{R}\to \mathcal{D}_{G}$ from the set of all RDs to the set of all GDs such that : for each complete relational database $D_{R}$ = $(S_{R}, I_{R})$,  $CM(DR)$ generates a complete graph database $D_{G} =(S_{G}, I_{G})$. $\square$
\end{definition}

In order to produce a complete graph database, our \emph{CM} process is based on two steps, \emph{schema mapping} and \emph{instance mapping}, which we detail hereafter.

\subsection{Schema Graph and Instance Graph}

Contrary to relational data, graph data still have no schema definition standard. Hence, we extend the property graph definition in order to introduce our schema graph definition.

\begin{definition}[\emph{Schema Graph}]
\sloppy{A \emph{schema graph} is an extended property graph defined with $S_{_G}=(V_{_S},E_{_S},L_{_S},A_{_S},Pk,Fk)$ where: 1) $V_{_S}$ is a finite set of vertices; 2) $E_{_S}\subseteq V_{_S}\times V_{_S}$ is a finite set of directed edges where $(v,v^{'})\in E_{_S}$ is an edge starting at vertex $v$ and ending at vertex $v^{'}$; 3) $L_{_S}$ is a function that assigns a label to each vertex in $V_{_S}$ (resp. edge in $E_{_S}$); 4) $A_{_S}$ is a function assigning a nonempty set of pairs ($a_i:t_i$) to each vertex (resp. edge) where $a_i\in \mathcal{A}$ and $t_i\in\mathcal{T}$; 5) \empty{Pk} is a partial function that assigns a subset of $A_s(v)$ to a vertex $v$; and finally 6) for each edge $e \in E_{_S}$, $Fk(e, s)$ (resp. $Fk(e, d)$) is a subset of $A_{_S}(e.s)$ (resp. $A_{_S}(e.d)$). $\square$}
\end{definition}

The functions $Pk$ and $Fk$ will be used later to incorporate integrity constraints over graph databases.

\begin{definition}[\emph{Instance Graph}]
\sloppy{Given a schema graph $S_{_G}=(V_{_S},E_{_S},L_{_S},A_{_S},Pk,Fk)$, an instance $I_G$ of $S_G$, called an \emph{instance graph}, is given by a property graph $I_{_G}=(V_{_I},E_{_I},L_{_I},A_{_I})$ where:}

\begin{enumerate}
\item $V_{_I}$ and $E_{_I}$ are the set of vertices and the set of edges as defined for schema graph;
\item for each vertex $v_{i} \in V_{_I}$, there exists a vertex $v_{s}\in V_{_S}$ such that: a) $L_{_I}(v_{i})=L_{_S}(v_{s})$; and b) for each pair $(a:c)\in A_{_I}(v_{i})$ there exists a pair $(a:t)\in A_{_S}(v_{s})$ with \emph{type}($c$)=$t$. We say that $v_{i}$ corresponds to $v_{s}$, denoted by $v_{i} \sim v_{s}$.
\item for each edge $e_{i}=(v_{i},w_{i})$ in $E_{_I}$, there exists an edge $e_{s}=(v_{s},w_{s})$ in $E_{_S}$ such that: a) $L_{_I}(e_{i})=L_{_S}(e_{s})$; b) for each pair $(a:c)\in A_{_I}(e_{i})$ there exists a pair $(a:t)\in A_{_S}(e_{s})$ with \emph{type}($c$)=$t$; and c) $v_{i} \sim v_{s}$ and $w_{i}\sim w_{s}$. We say that $e_{i}$ corresponds to $e_{s}$, denoted by $e_{i} \sim e_{s}$. $\square$
\end{enumerate}
\end{definition}

As for relational schema, a schema graph determines the structure, meta-information and typing that instance graphs must satisfy. It is clear that an instance $I_G$ of $S_G$ assigns a set of vertices (resp. edges) to each vertex $v_{s}$ (resp. edge $e_{s}$) in $S_G$ that have the same label as $v_{s}$ (resp. $e_{s}$). Moreover, a vertex $v_{i}$ (resp. edge $e_{i}$) in $I_G$ corresponds to a vertex $v_{s}$ (resp. edge $e_{s}$) in $S_G$ if the value $c$, attached to any attribute $a$ of $v_{i}$ (resp. edge $e_{i}$), respects the type $t$ given for $a$ within $v_{s}$ (resp. $e_{s}$).

\begin{figure}[t]
\centering
\includegraphics[width=\textwidth]{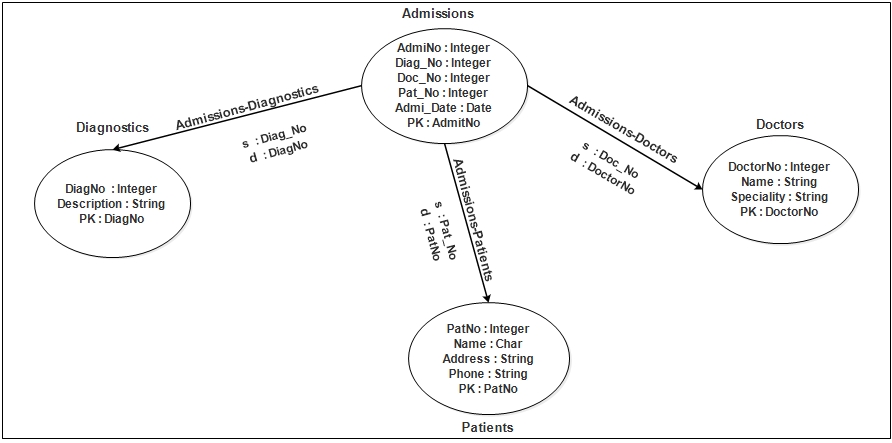}
\caption{Example of schema graph.}
\label{SG_IG}
\end{figure}

\begin{example}
\sloppy{Fig.\ref{SG_IG} depicts an example of a schema graph where each vertex (resp. edge) is represented naturally with its label (e.g. vertex \emph{Admissions}, edge \emph{Admissions-Doctors}) and a list of typed attributes (e.g. \emph{AdmiNo:Integer}). As a special case, the value of the attribute $Pk$ on some vertex refers to the value of the function $Pk$ on this vertex (e.g. $Pk$:$AdmiNo$ on vertex $Admissions$). Moreover, the values of attributes $s$ and $d$ on some edge $e$ refer to the values of the function $Fk(e,s)$ (resp. $Fk(e,d)$) at this edge (e.g. $s:Doc\_No$ and $d:DoctorNo$ on edge $Admissions-Doctors$). The use of these special attributes ($Pk$, $s$ and $d$) will be detailed later. One can see that the data graph of Fig.\ref{EX_PG} is an instance graph of the schema graph of Fig.\ref{SG_IG} since each vertex (resp. edge) of the former corresponds to some vertex (resp. edge) of the  latter. $\square$}
\end{example}

Finally, a \emph{graph database} is defined with $D_G=(S_G,I_G)$ where $S_G$ is a schema graph and $I_G$ is an instance of $S_G$. 

\subsection{Schema Mapping (\emph{SM})}

Given a relational schema $S_{R}=(R, A, T, \Sigma)$, we propose a schema mapping (\emph{SM}) process that produces a schema graph $S_{_G}=(V_{_S},E_{_S},L_{_S},A_{_S},Pk,Fk)$  as follows:

\begin{enumerate}

\item For each relation name $r \in R$ in $S_{R}$, there exists a vertex $v_{r} \in V_{_S} $ such that $L_{_S}(v_{r}) = r$;

\item For each $a \in A(r)$ with $T(a)=t$, we have a pair $(a:t) \in A_S(v_r)$;

\item For each \emph{primary key} $r[a_{1},\cdots,a_{n}] \in \Sigma$, we have $Pk(v_{r})$ = ``$a_{1},…,a_{n}$";

\item For each \emph{foreign key} $r[a_{1},\cdots,a_{n}] \to s[b_{1},\cdots,b_{n}] \in \Sigma$, we have an edge  $e = (v_{r},v_{s}) \in E_{_S}$ such that : a) $L_{_S}(e) = (s - r)$; b) $Fk(e,s)$= ``$a_{1},…,a_{n}$"; and c) $Fk(e,d)$ =  ``$b_{1},…,b_{n}$".

\item A special pair $(vid$:$Integer)$  is attached to each vertex of $S_G$ for storage concerns.
\end{enumerate}

\begin{figure}[t]
\begin{center}
\makebox[\textwidth][c]{\includegraphics[width=1.3\textwidth]{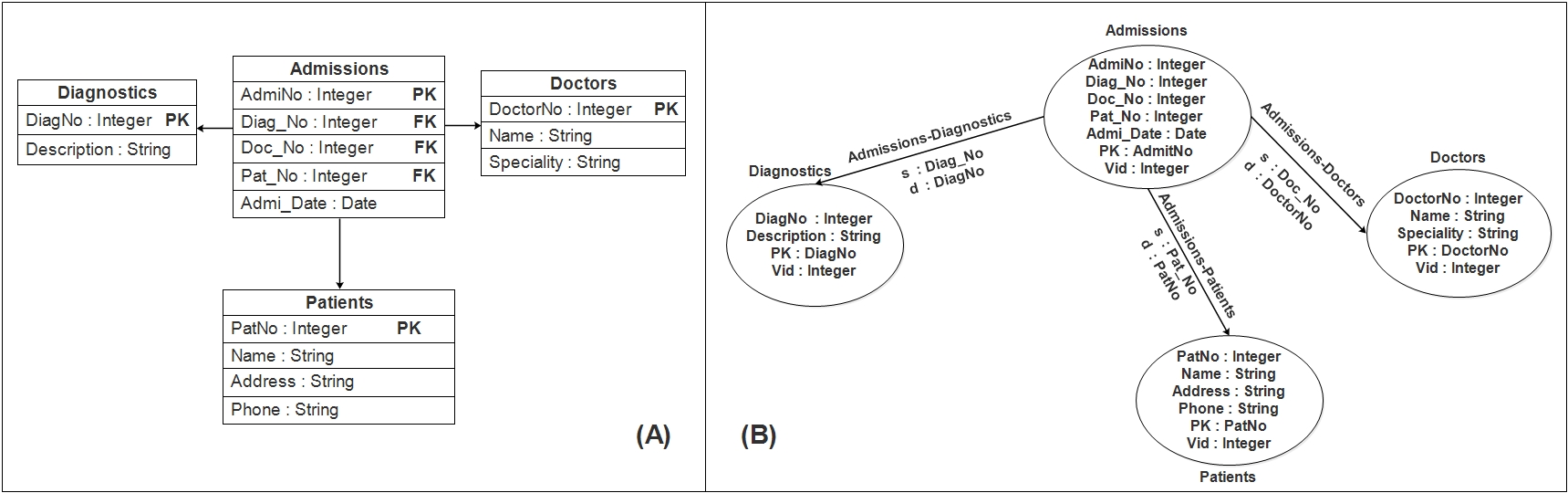}}%
\centering
\caption{Example of schema mapping.}
\label{ex3}
\end{center}
\end{figure}

\begin{example}
Fig.\ref{ex3} depicts ($A$) a relational schema and ($B$) its corresponding schema graph. One can see that our schema mapping rules are respected: $1$) each relation is mapped to a vertex that contains the label of this relation, its primary key and a list of its typed attributes; $2$) each foreign key between two relations (e.g. relations \emph{Admissions} and \emph{Patients} in part $A$) is represented by an edge between the vertices of these two relations (e.g. edge \emph{Admissions-Patients} in part $B$). $\square$
\end{example}

\subsection{Instance Mapping (\emph{IM})}\label{section_IM}

Given a relational database  $D_{R}$ = $(S_{R}, I_{R})$, we propose an instance mapping (\emph{IM}) process that maps the relational instance $I_{R}$ into an instance graph $I_{_G}=(V_{_I},E_{_I},L_{_I},A_{_I})$ as follows :

\begin{enumerate}

\item For each tuple $t\in I(r)$, there exists a vertex $v_t \in V_I$ with $L_I(v_t)=r$. We denote by $v_t$ the vertex that corresponds to the tuple $t$;

\item For each tuple $t\in I(r)$ and each attribute $a$ with $t(a)=c$, we have: $(a:c)\in A_I(v_t)$ if $a\neq tid$; and $(vid:c)\in A_I(v_t)$ otherwise.

\item for each \emph{foreign key} $r[a_{1},\cdots,a_{n}] \to s[b_{1},\cdots,b_{n}]$ defined with $S_G$ and any tuples $t\in I(r)$ and $t'\in I(s)$, if $t(a_i)=t'(b_i)$ for each $i\in[1,n]$, then: there is an edge $e=(v_{t},v_{t'})\in E_{_I}$ with $L_{_I}(e)=r-s$.
\end{enumerate}

\begin{example}
An example of our \emph{IM} process is given in Fig.\ref{ex2} where part (A) is the relational instance and part (B) is its corresponding instance graph. $\square$
\end{example}

\begin{figure}[t]
\centering
\makebox[\textwidth][c]{\includegraphics[width=1.3\textwidth]{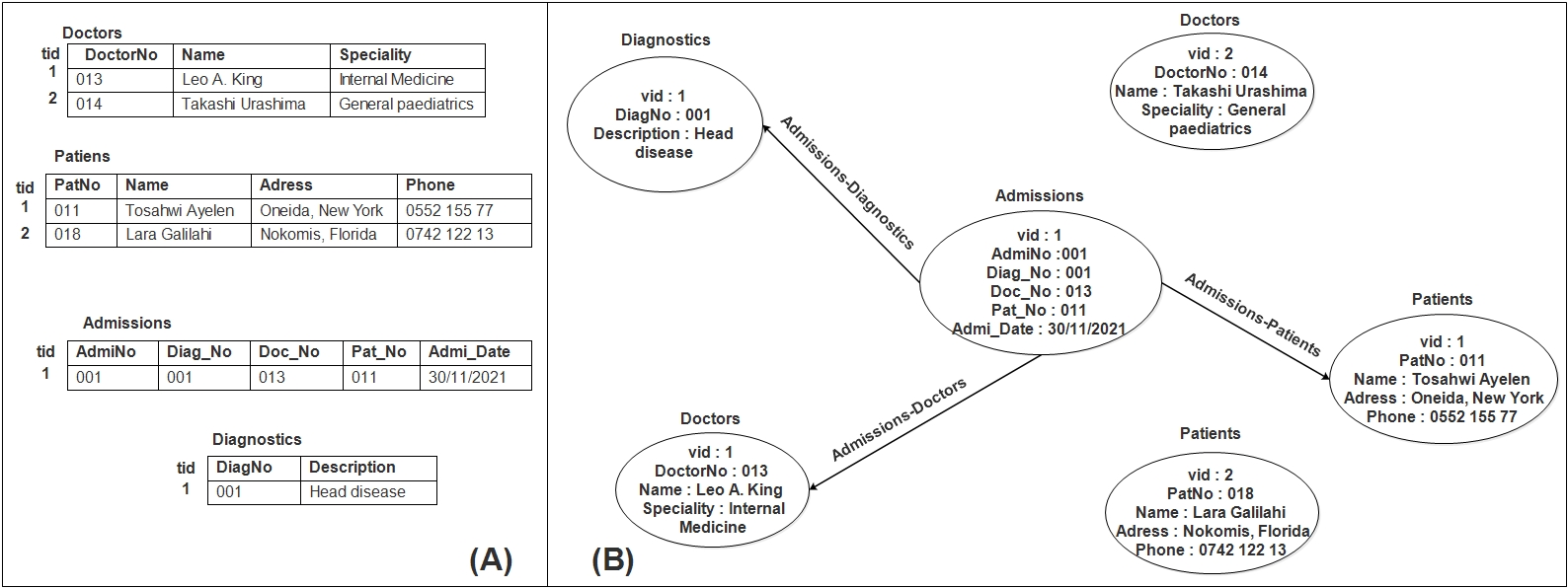}}%
\caption{Example of instance mapping.}
\label{ex2}
\end{figure}

It is clear that the special attribute $vid$ is used to preserve the tuples identification (i.e. the value of attribute $tid$) during the mapping process. 

We notice that our data mapping process is query language independent in the sense that any query language (e.g. Cypher, Gremlin, PGQL) can be applied over the resulting data graph.

\section{Properties of \emph{CM}}
We show that our \emph{CM} satisfies the three fundamental mapping properties \cite{Stoica2019a}: \emph{information preservation}, \emph{query preservation} and \emph{semantics preservation}.

\subsection{Information Preservation}
First, we explain that \emph{CM} does not lose any part of the information in the relational instance being translated :

\begin{theorem}
\label{t1}
The direct mapping \emph{CM}  is information preserving.
\end{theorem}

\begin{proof}
Theorem \ref{t1} can be proved easily by showing that there exists a computable mapping $CM^{-1}:D_G\rightarrow D_R$ that reconstructs the initial relational database from the generated graph database. Since our mapping $CM$ is based on two steps (schema and instance mappings), then $CM^{-1}$ requires the definition of $SM^{-1}$ and $IM^{-1}$ processes. Due to the space limit, the definition of $CM^{-1}$ is given in Appendix \ref{appendix_proof_th1}.
\end{proof}

\subsection{Query Preservation}
Second, we show that the way \emph{CM} maps relational instances into instance graphs allows one to answer the SQL query over a relational instance by translating it into an equivalent Cypher query over the generated  graph instance.

\begin{theorem}\label{t2}
The direct mapping \emph{CM}  is query Preserving.
\end{theorem}

\begin{proof}
Proving Theorem \ref{t2} can be done by providing an algorithm $S2C$ that, for any SQL query $Q_s$, produces an equivalent Cypher query $Q_c$ such that: for any relation database $D_R=(S_R,I_R)$ and any SQL query $Q_s$ over $I_R$, each row in $[Q_s]^{*}_{I_R}$ can be mapped to a row in $[Q_c]^{*}_{I_G}$ where $I_G\in CM(D_R)$ and $Q_c=S2C(Q_s)$. Our algorithm $S2C$ is summarized in Algorithm \ref{S2C}. Given an SQL query $Q_s$ in input, $S2C$ proceeds as follows: a) analyze and extract clauses from $Q_s$; b) compute for any SQL clause their equivalent Cypher clause; c) combine the resulted Cypher clauses in $Q_c$; and d) return the final Cypher query $Q_c$. Due to space limit, we give a reduced version of our query mapping algorithm that deals only with simple queries. However, one can easily extend our algorithm to deal with composed versions (e.g. queries with $IN$ clause). A running example of query mapping is given in Appendix \ref{appendix_qm}.

 \begin{algorithm}[t]
 \scriptsize
    \caption{S2C Algorithm }\label{S2C}
    \hspace*{\algorithmicindent} \textbf{Input: A simple SQL query $Q_s$, A relational schema $S_R$} \\
    \hspace*{\algorithmicindent} \textbf{Output: Its equivalent Cypher query $Q_c$.}
    \begin{algorithmic}[1]
    \State Rename relations (via $AS$-clause) in $Q_s$ if not applied;
    \State Extract the Select-clause ($SC$), the From-clause ($FC$) and the Where-clause ($WC$) from $Q_s$;
    \State Generate a Match-clause ($mc$) from $FC$:
    \Statex $a$) by translating each relation name $r$ in $SC$ into a vertex with label $r$; and
    \Statex $b$) by representing each join in $Q_s$ with an edge in $mc$ basing on $S_R$;
    
    \State Generate a Where-clause ($wc$) from $WC$ by translating each condition over a relation (attribute) in $WC$ into a condition over the corresponding vertex (attribute);
    
    \State Generate a Return-clause ($rc$) from $SC$ by translating each relation (attribute) in $SC$ into a corresponding vertex (attribute);
    
    \State Generate a Cypher query $Q_c$ by combining $mc$, $wc$ and $rc$;
    \State  Return $Q_c$;
    \end{algorithmic}
\end{algorithm}
\end{proof}

\subsection{Semantic Preservation}

Finally, we show that \emph{CM} is semantic preserving by checking consistency (resp. inconsistency) of relational database and graph database. Recall that a direct mapping is semantic preserving if any consistent (resp. inconsistent) relational database is translated into a consistent (resp. inconsistent) graph database. While consistency of relational database is well-known, no definition is given for graph database since there exists no standard for (schema) graph definition. To overcome this limit, we added the functions $Pk$ and $Fk$ to our schema graph definition in order to make possible the consistency checking for graph databases.

\begin{definition}[\emph{Graph consistency}]\label{definition_graph_consistency}
For any graph database $D_G=(S_G, I_G)$, the instance $I_G$ is said to be consistent w.r.t $S_G$ if:

\begin{itemize}
\item For each vertex $v_s\in V_s$ with $Pk(v_s)$=$``a_1,...,a_n"$ and each vertex $v_i\in V_I$ that corresponds to $v_s$: there exists no pair $(a_i:NULL) \in A_I(v_i)$ with $i \in [1,n]$. Moreover, for each $v'_i\in V_I\backslash\{v_i\}$ that corresponds to $v_s$, the following condition must not hold: for each $i \in [1,n]$, $(a_i:c) \in A_I(v_i) \cap A_i(v'_i)$.

\item For each edge $e_s\in E_s$ with $Fk(e_s,s)$=$``a_1,...,a_n"$ and $Fk(e_s, d)$=$``b_1,...,b_n"$, if $e_i=(v_1,v_2) \in E_I$ is an edge that corresponds to $e_s$ then we have: $(a_i:c)\in A_I(v_1)$ and $(b_i:c)\in A_I(v_2)$ for each $i\in [1,n]$. $\square$
\end{itemize}
\end{definition}

Intuitively, the consistency of graph databases is inspired from that of relational databases.

\begin{theorem}
\label{t3}
The direct mapping \emph{CM}  is semantic Preserving.
\end{theorem}

\begin{proof}

The proof of Theorem \ref{t3} is straightforward  and can be done by contradiction based on the mapping rules of \emph{IM} (Section \ref{section_IM}). Given a relational database $D_R = (S_R, I_R)$ and let $D_G = (S_G, I_G)$ be its equivalent graph database generated by \emph{CM}. We suppose that \emph{CM} is not semantic preserving. This means that either (A) $I_R$ is consistent and $I_G$ is inconsistent; or (B) $I_R$ is inconsistent while $I_G$ is consistent. We give only proof of case (A) since that of case (B) can be done in a similar way.

\noindent We suppose that $I_R$ is consistent w.r.t $S_R$ while $I_G$ is inconsistent w.r.t $S_G$. Based on Def. \ref{definition_graph_consistency} $I_G$ is inconsistent if one of the following conditions holds:

\noindent \textit{\textbf{1)}} There exists a vertex $v_i\in V_I$ that corresponds to a vertex $v_s\in V_s$ where: a) $Pk(v_s) = ``a_1,...,a_n"$; and b) $(a_i:NULL)\in A_I(v_i)$ for some attribute $a_{i\in[1,n]}$. Based on mapping rules of \emph{IM} process, $v_i$ corresponds to some tuple $t$ in $I_R$ and attributes $a_1,...,a_n$ correspond to a primary key defined over $I_R$ by $S_G$. Then (b) implies that the tuple $t$ assigns a $NULL$ value to the attribute $a_i$ which makes $I_R$ inconsistent. However, we supposed that $I_R$ is consistent.

\noindent \textit{\textbf{2)}} There are two vertices $v_1,v_2\in V_I$ that correspond to a vertex $v_s\in V_s$ where: a) $Pk(v_s) = ``a_1,...,a_n"$; and b) $(a_i:c)\in A_I(v_1)\cap A_I(v_2)$ for each $i\in[1,n]$. Based on mapping rules of \emph{IM} process, $v_1$ (resp. $v_2$) corresponds to some tuple $t_1$ (resp. $t_2$) in $I_R$ and attributes $a_1,...,a_n$ correspond to a primary key defined over $I_R$ by $S_G$. Then (b) implies that the tuples $t_1$ and $t_2$ assign the same value to each attribute $a_i$ which makes $I_R$ inconsistent. However, we supposed that $I_R$ is consistent.

\noindent \textit{\textbf{3)}} There exists an edge $e_i=(v_1,v_2)$ in $E_I$ that corresponds to an edge $e=(v_s,v_d)$ in $E_s$ where: a) $Fk(e,s) = ``a_1,...,a_n"$; b) $Fk(e, d) = ``b_1,...,b_n"$; and c) there exists some attribute $a_{i\in[1,n]}$ with $(a_i:c_1)\in A_I(v_1)$, $(a_i:c_2)\in A_I(v_2)$, and $c_1\neq c_2$. Based on mapping rules of \emph{IM} process, $v_1$ (resp. $v_2$) corresponds to some tuple $t_1$ (resp. $t_2$) in $I_R$, $v_s$ (resp. $v_d$) corresponds to some relation $s$ (resp. $d$) in $S_R$, the function $Fk$ over edge $e$ refers to a foreign-key $s[a_1,...,a_n]\rightarrow d[b_1,...,b_n]$ defined over $I_R$ by $S_G$. Then (b) implies that the tuple $t_1$ assigns a value to some attribute $a_{i\in[1,n]}$ that is different to that assigned by tuple $t_2$ to attribute $b_{i\in[1,n]}$. This means that there is a violation of foreign-key by tuple $t_1$ which makes $I_R$ inconsistent. However, we supposed that $I_R$ is consistent.\\

\noindent Therefore, each case of $I_G$ inconsistency leads to a contradiction, which means that if $I_R$ is consistent then its corresponding $I_G$ cannot be inconsistent.

By doing proof of part (B) in a similar way, we conclude that if $I_R$ is consistent (resp. inconsistent) then its corresponding instance graph $I_G$ must be consistent (resp. inconsistent). This completes the proof of Theorem \ref{t3}.$\square$
\end{proof}

\section{Conclusion and Future works}
In this paper, we proposed a complete mapping process that translates any relational database into an equivalent graph database by considering both schema and data mapping. Our mapping preserves the information and semantics of the relational database and maps any SQL query, over the relational database, into an equivalent Cypher query to be evaluated over the produced graph database. We plan to extend our model to preserve another mapping property, called monotonicity \cite{Stoica2019a}, which ensures that any update to the relational database will not require generating the corresponding graph database from scratch. Also, we will enrich the definition of the relational schema with more integrity constraints. We are conducting an experimental study based on real-life databases to check our approach's efficiency and effectiveness.

\newpage
\appendix
\noindent\textbf{\textsc{APPENDIX}}
\section{Proof of Theorem \ref{t1}}\label{appendix_proof_th1}
We prove Theorem \ref{t1} by providing a computable mapping $CM^{-1}:D_G\rightarrow D_R$ that reconstructs the original relational database from the generated graph database. To do so, $CM^{-1}$ must consist of two processes: 1) $SM^{-1}$ that reconstructs the original relational schema from the generated graph schema; and 2) $IM^{-1}$ that reconstructs the original relational instance from the generated instance graph. We provide hereafter complete definitions of $SM^{-1}$ and $IM^{-1}$.\\

\noindent\textbf{$SM^{-1}$ transformation rules.}
\sloppy{Given a schema graph $S_{_G}=(V_{_S},E_{_S},L_{_S},A_{_S},Pk,Fk)$ , the reversed schema mapping $SM^{-1}$ produces  a relational schema $S_{R}=(R, A, T, \Sigma)$ as follows:}

\begin{itemize}
\item For each vertex $v_r\in V_s$, we create a relation $r \in R$ with name $L_s(v_r)$. We denote by $v_r$ the vertex that corresponds to the relation $r$.

\item For each pair $(a:t) \in A_s(v_r)$, we add the attribute $a$ to $A(r)$ with $T(a)=t$.

\item For each $Pk(v_r)="a_1,...,a_n"$, we add $r[a_1,..,a_n]$ to the set $\Sigma$.

\item For each edge $e=(v_r, v_s)$ with $Fk(e,s)$=$``a_1,...,a_n"$ and $Fk(e,d)=``b_1,...,b_n"$, we add $r[a_{1},\cdots,a_{n}] \to s[b_{1},\cdots,b_{n}]$ to the set $\Sigma$.
\end{itemize}

Notice that the special pair $(vid:Integer)$ is ignored by the mapping $SM^{-1}$ since it has no equivalent part in the relational schema. However, it allows instance graphs to preserve the tuples identification.\\

\noindent\textbf{$IM^{-1}$ transformation rules.}
Given an instance graph $I_G=(V_I, E_I, L_I, A_I)$, the reversed instance mapping $IM^{-1}$ produces the original relational instance $I_R$ as follows:

\begin{itemize}
\item For each vertex $v\in V_I$ with $L_I(v)=r$, we create a tuple $t \in I(r)$ such that:
\subitem $t(tid)=c$ where $(vid:c)\in A_I(v)$; and
\subitem $t(a)=c$ for each $(a:c)\in A_I(v)$.
\end{itemize}

From the definition of \emph{IM} mapping (Section \ref{section_IM}), one can verify that for any two vertices $v_1,v_2\in V_I$, if $L_I(v_1)=L_I(v_2)$ then $v_1$ and $v_2$ have different values for attributes $vid$. Based on this remark, tuples generated by $IM^{-1}$ mapping satisfy the following condition: for any two tuples $t_1,t_2\in I(r)$, $t_1(tid)\neq t_2(tid)$.

\section{Running example of query mapping by S2C \ref{S2C} }\label{appendix_qm}
\begin{example}
\begin{figure}[t]
\centering
\includegraphics[width=12cm]{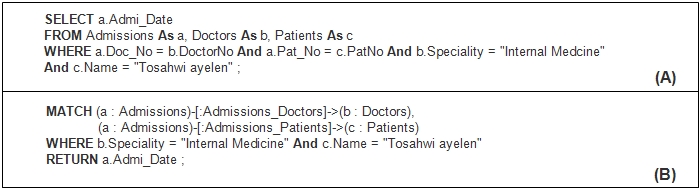}
\caption{Example of (A) an SQL query $Q_s$ and (B) its equivalent Cypher query $Q_c$.}
\label{q1}
\end{figure}

\begin{figure}[t]
\centering
\includegraphics[width=12cm]{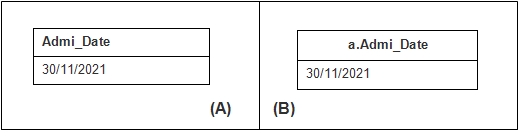}
\caption{Table (A) the result of $Q_s$  over $I_R$ and (B) the result of its equivalent $Q_c$ over $I_G$.}
\label{q2}
\end{figure}

Fig. \ref{q1} represents (A) an SQL query $Q_s$ and (B) its equivalent Cypher query $Q_c$. It is clear to see that each SQL clause is translated into a Cypher clause as explained by algorithm $S2C$. Precisely, each relation name (\emph{Admissions}, \emph{Doctors}, \emph{Patients}) in \emph{From-Clause} is translated into a vertex in $mc$ labeled with name of this relation. Moreover, each join operation between two relations is mapped to an edge between vertices corresponding to these relations. Each condition (e.g. $a.Doc\_No = b.DoctorNo$) in \emph{Where-Clause} is translated into an equivalent condition in $wc$. Finally, the selected items to return are extracted from the \emph{Select-clause}. The final Cypher query $Q_c$ is composed by combining the three clauses \emph{Match} ($mc$), \emph{Where} ($wc$) and \emph{Return} ($rc$).

Consider the relational instance $I_R$ of Fig.\ref{ex2} (A) and its equivalent instance graph $I_G$ in Fig.\ref{ex2} (B). When we apply $Q_s$ over $I_R$ (resp. $Q_c$ over $I_G$), we obtain the result table depicted in Fig.\ref{q2} (A) (resp. (B)). It is easy to see that both $Q_s$ and $Q_c$ return the same information. $\square$
\end{example}

\end{document}